\documentclass[a4paper,USenglish,cleveref,numberwithinsect,pdfa]{lipics-v2021}
\pdfoutput=1

\usepackage{microtype}
\usepackage[T5,T1]{fontenc}

\title{Two or three things I know about tree transducers}

\author{Lê Thành Dũng (Tito) Nguy\~{ê}n}{École normale supérieure de Lyon, France \and \url{https://nguyentito.eu/}}{nltd@nguyentito.eu}{https://orcid.org/0000-0002-6900-5577}{}

\authorrunning{L.~T.~D.~Nguy\~{ê}n}

\Copyright{}



\category{}

\hideLIPIcs
\nolinenumbers

\EventEditors{}
\EventNoEds{0} 
\EventLongTitle{}
\EventShortTitle{}
\EventAcronym{}
\EventYear{}
\EventDate{}
\EventLocation{}
\EventLogo{}
\SeriesVolume{}
\ArticleNo{}

\usepackage{mathtools}
\usepackage{csquotes}
\usepackage{mdframed}
\usepackage{tikz}
\usetikzlibrary{shapes,arrows,positioning,automata,shadows}
\usepackage{amsfonts,amsmath,amssymb}
\usepackage{amsthm}
\usepackage{hyperref}
\usepackage{stmaryrd}

\newcommand{\cC}{\mathcal{C}}
\newcommand{\SO}{\mathrm{SO}}
\newcommand{\MTT}{\mathrm{MTT}}
\newcommand{\MSOTS}{\mathrm{MSOTS}}

\newcommand{\unfold}{\mathtt{unfold}}
\newcommand{\yield}{\mathtt{yield}}
\newcommand{\lrangle}[1]{\langle #1 \rangle}

\newcommand{\qzero}{{\color{blue}q_0}}
\newcommand{\qone}{{\color{red}q_1}}

\theoremstyle{claimstyle}

\begin{document}
\maketitle

\begin{abstract}
  You might know that the name \enquote{tree transducers} refers to various kinds of automata that compute functions on ranked trees, i.e.\ terms over a first-order signature. But have you ever wondered about:
  \begin{itemize}
    \item How to remember what a macro tree transducer does?
    \item Or what are the connections between top-down tree(-to-string) transducers, multi bottom-up tree(-to-string) transducers, tree-walking transducers, (invisible) pebble tree transducers, monadic second-order transductions, unfoldings of rooted directed acyclic graphs (i.e.\ term graphs) --- and what happens when the functions that they compute are composed?
  \end{itemize}
  The answers may be found in old papers (mostly coauthored by Engelfriet), but
  maybe you can save some time by first looking at this short note.
\end{abstract}

\tableofcontents

\section{Introduction}

This relatively short note exists for two reasons.

First, I was frustrated at having to repeat myself in multiple papers involving
tree transducers in some way, such as~\cite{Sandra}, to justify claims that
follow more-or-less immediately from \enquote{connecting the dots} in the
literature. This is why I wanted to record a bunch of useful facts in
\Cref{sec:relationships}. Some parts of that section are also more pedagogical
in purpose.

As for \Cref{sec:top-bot}, it is meant to spread the word about a
\enquote{moral} rather than technical point, namely the \enquote{bottom-up}
perspective on deterministic macro tree transducers, to dispel the latter's
reputation for being overly complicated. From the reactions I observed at the
Dagstuhl Seminar on Regular Transformations in May 2023~\cite{Dagstuhl23202}, it
seems that many people working on string transducers nowadays find this
perspective enlightening, due to its close proximity to the well-known streaming
string transducer~\cite{SST,CopyfulSST}. In fact, macro tree transducers are
already presented this way in Courcelle and Engelfriet's
book~\cite[Section~8.7]{courcellebook}; but I hope that my alternative
exposition will prove helpful for some people. Another point raised is
\enquote{when is a tree-to-string transducer not quite the same thing as a tree
  transducer that outputs a string?}. Finally, \Cref{sec:strongly-sur} is a
niche observation about something I found perplexing at first.

\section{Top-down states vs bottom-up registers}
\label{sec:top-bot}

\subsection{Deterministic top-down tree transducers are bottom-up}

A \emph{top-down tree transducer} is like a tree automaton, except that the
result of a transition is not just one state per child of the current node, but
a tree expression involving states applied to children. It can also be seen as a
regular tree grammar controlled by an input tree --- just like the derivations of
a grammar, the semantics of a top-down tree transducer can be defined by
rewriting.
\begin{example}[{\enquote{conditional swap}, inspired by~\cite[\S2.3]{STT}}]\label{ex:cond-swap}
  Consider the function
  \[ f(a(t,u)) = a(f(u),f(t)) \qquad f(t)=t\ \text{when the root of $t$ is not}\ a \]
  on trees over the ranked alphabet $\{a:2,\; b:1,\; c:0\}$.

  To compute it, we use the initial state $q_{0}$, an auxiliary state $q_{1}$ and the transitions:
  \[ \qzero\langle a(t,u)\rangle \to a(\qzero\langle u \rangle, \qzero\langle t \rangle) \qquad \qone\langle a(t,u)\rangle \to a(\qone\langle t \rangle, \qone\langle u \rangle) \qquad \qzero\langle b(t) \rangle \to b(\qone\langle t \rangle) \qquad  \dots \]
  They are \emph{deterministic}: each pair of a state and a ranked letter appears only once on the left-hand side of a rule.
  Over the input $a(b(c),c)$, we have the following executions:
  \[ \qzero\langle a(b(c),c) \rangle \to a(\qzero\langle c \rangle, \qzero\langle b(c) \rangle) \to a( \qzero\langle c \rangle, b(\qone\langle c \rangle)) \to \dots \to a(c,b(c)) \]
\[ \qzero\langle b(a(b(c),c)) \rangle \to b(\qone\langle a(b(c),c) \rangle) \to \dots \to b(a(b(c),c)) \]
  By determinism, the rewriting relation $\to$ always reaches a unique normal form. In this case, $\qzero\lrangle{t} \to^{*} f(t)$ and $\qone\lrangle{t} \to^{*} t$ for any input tree $t$.
\end{example}

\subparagraph{The bottom-up view.}

Another informal way to see a top-down tree transducer is that the states are
mutually recursive procedures. A well-established principle in algorithmics is
that a top-down structurally recursive procedure can be optimized into a
bottom-up \enquote{dynamic programming} algorithm. We can also adopt this point
of view here; as we shall see in later subsections, it brings us closer in
spirit to more recent work on transducers (cf.~\S\ref{sec:mtt}).

\begin{example}
  The previous example of top-down tree transducer can be seen as a bottom-up
  device whose memory consists of two tree-valued \emph{registers}
  $X_{0},X_{1}$. After processing a subtree $t$, the contents of the register
  $X_{i}$ is $q_{i}\lrangle{t}$ (for $i\in\{0,1\}$). The \enquote{final
    output register} $X_{0}$ corresponds to the initial state $q_{0}$.
  
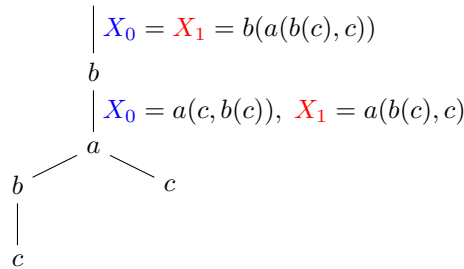
\begin{center}
      \begin{tikzpicture}
        \node (top) at (0,1) {};
        \node (mid) at (0,0) {$b$};
        \node (bot) at (0,-1) {$a$};
        \node (b) at (-1,-1.5) {$b$};
        \node (c) at (1,-1.5) {$c$};
        \node (cc) at (-1,-2.5) {$c$};
        \draw[-] (top) -- node[right]{${\color{blue}X_{0}} = {\color{red}X_{1}} = b(a(b(c),c))$} (mid);
        \draw[-] (mid) -- node[right]{${\color{blue}X_{0}} = a(c,b(c)),\; {\color{red}X_{1}} = a(b(c),c)$} (bot);
        \draw (bot) -- (b);
        \draw (bot) -- (c);
        \draw (b) -- (cc);
      \end{tikzpicture}
  \end{center}
\end{example}

\subparagraph{Beware of nondeterminism.}

Let's consider the following nondeterministic rules:
\[ q\lrangle{b(t)} \to a(q\lrangle{t},q\lrangle{t}) \qquad q\lrangle{c} \to c \qquad q\lrangle{c} \to b(c) \]
By taking the associated rewriting system, we have:
\[ q\lrangle{b(c)} \to^{*} a(c,c)\ \text{or}\ a(b(c),c)\ \text{or}\ a(c,b(c))\ \text{or}\ a(b(c),b(c)) \]
However, if we work bottom-up, here is what happens. After reading $c$, the register contains either $c$ or $b(c)$; this is chosen nondeterministically. Then, at the next step, after reading $b(c)$, we get either $a(c,c)$ or $a(b(c),b(c))$ depending on the earlier choice --- but not $a(c,b(c))$.
Morally, using the terminology of programming languages: \enquote{top-down nondeterminism is naturally call-by-name, bottom-up nondeterminism is naturally call-by-value} (this is related to \enquote{IO vs OI} in automata theory).

From now on, in this article, \textbf{all automata models under consideration will be deterministic}.

\subsection{Top-down lookahead = bottom-up states}

Transducers are often extended with a feature called \emph{regular
  lookahead}, introdu. In general, this notion corresponds in automata theory to being
able to choose a transition based on a regular property of the
\enquote{future part} of the input. Concretely, in the case of top-down tree
transducers~\cite{Engelfriet77,Engelfriet77Erratum}, the transitions are now of the form
\[ q\lrangle{a(t_{1}|r_{1},\dots,t_{k}|r_{k})} \to \text{some term e.g.}\ b(a(q'\lrangle{t_{3}}, c))  \]
where the $r_{i}$ are some states of an auxiliary \emph{deterministic bottom-up tree automaton}, called the lookahead automaton. The idea is that this transition can only be applied when the run of the lookahead automaton on the subtree $t_{i}$ ends in state $r_{i}$ (for all $i$). The definition of determinism is adjusted accordingly.
\begin{example}
  Suppose we want to replace all subtrees of the input of the form $b(t)$ where
  $t$ does not contain any $b$ by $a(t,t)$. This can be done with a single state $q$ and
  regular lookahead. The lookahead states are $r_\oplus$ if the subtree contains a
  $b$, and $r_{\ominus}$ otherwise (this can indeed be implemented by a deterministic
  bottom-up tree automaton). The transitions include:
  \[ q\lrangle{b(t|r_{\oplus})} \to b(q\lrangle{t}) \qquad q\lrangle{b(t|r_{\ominus})} \to a(q\lrangle{t},q\lrangle{t}) \]
\end{example}
From the bottom-up point of view, this regular lookahead just corresponds to
\emph{adding states} to the transducer! That is, a configuration of a
corresponding bottom-up device now consists not only of register contents, but
also of a finite control state; but the computation is still performed entirely
bottom-up.
  \begin{center}
      \begin{tikzpicture}
        \node (top) at (0,1) {};
        \node (mid) at (0,0) {$b$};
        \node (bot) at (0,-1) {$c$};
        \draw[-] (top) -- node[right]{$\text{state} = r_{\oplus},\; X = a(c,c)$} (mid);
        \draw[-] (mid) -- node[right]{$\text{state} = r_{\ominus},\; X = c$} (bot);
      \end{tikzpicture}
  \end{center}

The observations in this subsection and in the previous one are essentially the
meaning of a paper by Fülöp, Kühnemann and Vogler~\cite{FulopKV04}:
\enquote{deterministic top-down tree transducers with regular lookahead =
  deterministic \emph{multi bottom-up tree transducers}}.

\subsection{Top-down tree-to-string transducers --- it's not what you think(?)}
\label{sec:string-output}

One of the major results in transducer theory in the 2010s has been the
decidability of equivalence for \emph{top-down tree-to-string
  transducers}~\cite{SeidlMK18}, which popularized the \enquote{Hilbert method}
(cf.~\cite{HilbertMethod}). But what are these devices?

\subparagraph{Strings as unary trees.}

We start by presenting a --- arguably \emph{the} --- natural way to restrict a
tree transducer model to strings as output. But here, when applied to top-down
tree transducers, it does \emph{not} give us the right thing. It relies on the
following encoding: strings over the finite alphabet $\Sigma$ are in canonical
bijection with trees over the ranked alphabet with:
\begin{itemize}
  \item a unary letter for each letter of $\Sigma$;
  \item a 0-ary letter $\varepsilon$ that serves as an \enquote{end-of-string}
        symbol.
\end{itemize}
For example, the string $abac\in\{a,b,c\}^{*}$ becomes the unary tree
$a(b(a(c(\varepsilon))))$.

The tree-to-string functions that can be computed this way by top-down tree
transducers are very weak, since they can only explore one branch of the input
tree. Note, however, that if the input also consists of strings encoded as unary
trees, then we get the very well-studied model of \emph{sequential transducers},
see for instance~\cite[\S{}V.1.2]{Sakarovitch}.

\subparagraph{The yield operation and concatenable strings.}

In fact, in the older literature, it often happens that \enquote{$X$
  tree-to-string transducer} means $\yield \circ X\ \text{tree transducer}$ ---
and it is the case for $X =$ \enquote{top-down}. The \emph{yield} of a tree is
the string obtained by reading its leaves from left to right, and erasing some
letters deemed \enquote{neutral}.

Since
$\yield(a(t_{1},\dots,t_{k})) = \yield(t_{1}) \cdot \ldots \cdot \yield(t_{k})$,
this amounts to working with a data type of \emph{concatenable strings} for the
output. That is, we can describe top-down tree-to-string transducers directly
(without $\yield$) using transitions whose right-hand side may use string
concatenation (as in context-free grammars).
\begin{example}
  The following single-state top-down tree-to-string transducer computes the
  postfix representation (\enquote{reverse Polish notation}) of its input tree:
  \[ q\lrangle{a(t,u)} \to q\lrangle{t} \cdot q\lrangle{u} \cdot a \qquad q\lrangle{b(t)} \to q\lrangle{t} \cdot b \qquad q\lrangle{c} \to c \]
\end{example}

From the bottom-up point of view, this transducer model should be equivalent to
some counterpart to Fülop et al.'s multi bottom-up tree
transducer~\cite{FulopKV04} whose registers contains concatenable strings.
Indeed, the \emph{multi bottom-up tree-to-string transducer} has been studied
--- under this precise name --- in Courcelle and Engelfriet's
book~\cite[\S8.6]{courcellebook}. As before, top-down regular lookahead
corresponds to bottom-up states.

\subparagraph{The string-to-string case.}

It does not make much sense for an automaton to process its input as
concatenable strings. Thus, to get a sensible string-to-string version of
top-down tree-to-string transducers, we should take unary trees as input, but
still use concatenation to produce output strings. The resulting machine model
is a variant of the \emph{copyful streaming string transducer}~\cite{CopyfulSST}
--- whose memory contains a finite state and concatenable registers --- that reads
its input from right to left, instead of left to right.

For other characterizations of this class of string-to-string functions,
see~\cite{FerteMarinSenizergues}.
\begin{remark}
  Arguably, since concatenation provides the canonical semigroup structure on
  strings, the \emph{recognition of a regular language by a semigroup morphism}
  processes its input as concatenable strings. An algebraic characterization of
  string-to-string MSO transductions in this vein has been obtained
  in~\cite{BojanczykN23}. Morally, the fact that both input and output are
  treated as concatenable strings makes it easy to show closure under
  composition.
\end{remark}

\subsection{Macro tree transducers store tree contexts}
\label{sec:mtt}

Traditionally, macro tree transducers~\cite{Macro} (MTTs) are presented as an extension of top-down tree transducers with \emph{parameters}, whose rules may look like, for example,
\[ \qzero\lrangle{a(t,u)} \to \qone\lrangle{t} (b(\qzero\lrangle{u})) \quad \qone\lrangle{a(t,u)}(x) \to \qone\lrangle{u}(\qone\lrangle{u}(x)) \quad \qone\lrangle{c}(x) \to a(x,x) \quad\dots\]
where $\qzero$ has no parameters and $\qone$ has the parameter $x$. A possible run of this MTT is
\[ \qzero\lrangle{a(a(b(c),c),c)} \to \qone\lrangle{a(b(c),c)}(b(\qzero\lrangle{c})) \to   \qone\lrangle{c}(\qone\lrangle{c}(b(\qzero\lrangle{c}))) \to \dots \]
In such a run, $\qone\lrangle{\text{something}}$ is always applied to an expression that eventually reduces to a tree. But one could also evaluate $\qone\lrangle{t}(x)$ for a formal parameter $x$, for instance:
\[ \qone\lrangle{a(c,c)}(x) \to \qone\lrangle{c}(\qone\lrangle{c}(x)) \to \qone\lrangle{c}(a(x,x)) \to a(a(x,x),a(x,x))\]
We may therefore say that the \enquote{value} of $\qone\lrangle{a(c,c)}$
is this \enquote{tree with parameters at some of the leaves}, usually called a
\emph{tree context}. It represents the map $x \mapsto a(a(x,x),a(x,x))$.

Thus, the \enquote{bottom-up view} is as follows: \emph{deterministic macro tree
  transducers can be seen as bottom-up devices whose memory consists of
  registers storing tree contexts}. This perspective in presented in
Courcelle--Engelfriet~\cite[Section~8.7]{courcellebook}; it makes the MTT is an
ancestor to later transducer models that are explicitly presented as
manipulating context-valued registers, such as the streaming tree transducer
of~\cite{STT} and the register tree transducer of~\cite[Section~4]{FOTree}.

\subparagraph{Lookahead elimination.}

As in the case of top-down tree(-to-string) transducers, regular lookahead
corresponds to bottom-up states in macro tree transducers. That said, the
significance of the feature differs between these two cases. For deterministic
top-down tree transducers, deciding whether a given transducer with lookahead
computes a function that could be computed without lookahead is still an open
problem, see~\cite{EngelfrietMS16} for partial results. For deterministic MTTs,
the answer is always \enquote{yes}:
\begin{theorem}[{\cite[Theorem~4.21]{Macro}}]
  For any deterministic MTT with regular lookahead, there is another
  deterministic MTT without lookahead that computes the same function.
\end{theorem}

To see why this is true, first, it is important to note that MTTs allow states
of arbitrary arities, i.e.\ with an arbitrary number of parameters. For instance
one may have a ternary state $q_{2}$ with a rule such as
$q_{2}\lrangle{b(t)}(x_{1},x_{2},x_{3}) \to a(q_{1}\lrangle{t}(x_{3}), x_{1})$.
We can then encode the bottom-up \enquote{lookahead states}, ranging over a
finite set $R = \{r_{1},\dots,r_{n}\}$, in the value of an $n$-ary state, i.e.\
a bottom-up register storing an $n$-ary context. The idea --- analogous to the
Church encoding of finite sets in $\lambda$-calculus --- is to represent the
lookahead state $r_{i}$ by the context $(x_{1},\dots,x_{n}) \mapsto x_{i}$.

\subsection{MTTs and string outputs --- it is and isn't what you think}
\label{sec:macro-string}

Following \Cref{sec:string-output}, there are two ways to use macro tree
transducers to output strings.

\subparagraph{Unary trees.}

Let us first look into the encoding of output strings as unary trees. The key
remark is that \emph{concatenable strings can be represented by contexts}, with
concatenation implemented by composition:
\[ ab \cdot ac = abac \quad\rightsquigarrow\quad (x\mapsto a(b(x))) \circ (x\mapsto a(c(x))) = (x\mapsto a(b(a(c(x)))))\]
Conversely, over a ranked alphabet of unary letters plus an end-of-string marker
$\varepsilon$, any context with $n$ parameters $x_{1},\dots,x_{n}$ must have the
form
\[ \qquad\qquad\qquad\underbrace{a_{1}(a_{2}(\dots (a_{n}}_{\mathclap{\text{can be represented by the string}\ a_{1}\dots a_{n}}}(\underbrace{\text{either}\ x_{i}\ \text{for some}\ i\in\{1,\dots,n\}\ \text{or}\ \varepsilon}_{\text{finite data}})) \dots )) \]
Therefore, we have:
\begin{theorem}[{\cite[Lemma~7.6]{MacroMSO}}]\label{thm:macro-top-down}
  Macro tree transducers with unary trees as output are equivalent in expressive
  power to top-down tree-to-string transducers with regular lookahead.
\end{theorem}
According to the end of \Cref{sec:string-output}, this means that macro tree
transducers whose inputs and outputs both consist of unary trees are the same as
right-to-left copyful streaming string transducers (SSTs). Note that states can
be eliminated in copyful SSTs~\cite[Corollary~3.6]{CopyfulSST}, but for reasons
that crucially depend on the input being a string; these reasons are unrelated
to the elimination of lookaheads in MTTs sketched earlier.

\begin{remark}\label{sec:strongly-sur}
  Right-to-left \emph{copyless} SSTs~\cite{SST} would then correspond to
  \emph{strongly single-use} MTTs~\cite{MacroMSO} with regular lookahead whose
  inputs and outputs are unary trees, with a subtlety: in MTTs, the output is
  taken to be the final value of some fixed state/register, whereas in SSTs, the
  output depends on the final register values in a less restrictive way. This
  makes a difference as the function below can be computed by a right-to-left
  copyless SST (exercise for the reader), but not by a strongly single-use
  MTT~\cite[proof of Theorem~5.6]{MacroMSO}:
  \[ a^{n} \mapsto a^{n} \qquad a^{n}bw \mapsto a^{n}bb^{|w|}\ \text{for}\ w \in \{a,b\}^{*} \]
\end{remark}

\subparagraph{Yield.}

Post-composing MTTs with the $\yield$ operation results in the so-called
\enquote{macro tree-to-string transducers}, which manipulate \enquote{string
  contexts} such as $(x_{1},x_{2}) \mapsto x_2 ab x_{1} ac x_{2} b$. Morally,
these transducers are to macro grammars~\cite{Fischer68} what top-down
tree-to-string transducers are to context-free grammars; this is where the name
comes from. Several equivalent characterizations of macro tree-to-string
transducers may be found in~\cite{EngelfrietPushdownMacro}. We shall discuss
them further in \Cref{sec:mtt-comp-string}.

\section{Relationships between deterministic transduction classes}
\label{sec:relationships}

\subsection{Monadic second-order transductions with sharing / unfolding}
\label{sec:unfolding}

\emph{Monadic Second-Order logic} (MSO) is a logic over relational structures
(sets endowed with relations). For instance one can encode a string $w$ as a
relational structure by taking its set of positions $\{1,\dots,|w|\}$ endowed
with the total order $\leqslant$ and, for each letter $a$, a unary relation
$a(i) =$ \enquote{the $i$-th letter is $a$}. Over strings (resp.\ trees), the
properties that MSO can define correspond exactly to the regular languages
(resp.\ regular tree languages). Thus, MSO provides a canonical generalisation
of regular languages to structures beyond strings and trees; typically, graphs,
as discussed in Courcelle and Engelfriet's book~\cite{courcellebook}.

\emph{MSO transductions}, also covered in~\cite{courcellebook}, are the usual
way to define transformations of relational structures using MSO. (One may also
consider MSO interpretations~\cite{msoInterpretations} or MSO set
interpretations~\cite{ColcombetL07} but they are not always as well-behaved,
especially over structures other than strings and trees.) String-to-string MSO
transductions are also called \emph{regular functions} and have a rich theory,
see the survey~\cite{MuschollPuppis}. Tree-to-tree MSO transductions are also
well-studied, with several equivalent characterizations by machine
models~\cite{AttributedMSO,MacroMSO,STT,FOTree}. For instance a major theorem of
Engelfriet and Maneth~\cite{MacroMSOLinear} is that the tree-to-tree functions
definable by MSO transductions are exactly those computed by macro tree
transducers of \emph{linear growth}, i.e.\ $\text{output size} = O(\text{input
  size})$; see \Cref{thm:inaba} for a generalisation. Thus, most examples of
tree-to-tree functions seen in the previous section are MSO transductions.

However, one can easily define top-down tree transducers with non-linear growth:
\begin{example}
  The following transducer
  \[ \qzero\lrangle{S(t)} \to a(\qone\lrangle{t},\qzero{\lrangle{t}}) \qquad \qzero\lrangle{0} \to c \qquad \qone\lrangle{S(t)} \to b(\qone\lrangle{t}) \qquad \qone\lrangle{0} \to b(c) \]
  computes the function
  $S^{n}(0) \mapsto a(b^{n}(c),a(b^{n-1}(c), \dots a(b(c),c)\dots))$ with quadratic
  growth.
\end{example}

The thing to note about this example is that the output contains many repeated
subtrees. For instance, for $S(S(0)) \mapsto a(b(b(c)),a(b(c),c))$, the
subtree $b(c)$ is repeated twice; and this is because in the computation,
$\qone\lrangle{0}$ appears at two different places. One may compress the output
tree into a \enquote{shared representation} --- a \emph{rooted directed acyclic
  graph} (DAG) --- where all copies of each subtree produced by the same
$q\lrangle{t}$ are merged together:
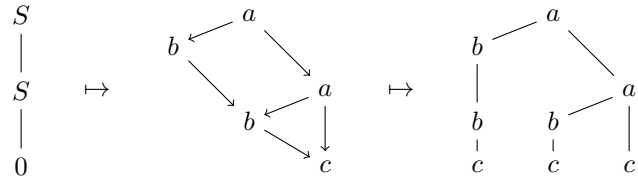
\begin{center}
  \begin{tikzpicture}
    \node (a) at (0,0) {$S$};
    \node (b) at (0,-1) {$S$};
    \node (c) at (0,-2) {$0$};
    \draw (a) -- (b);
    \draw (b) -- (c);

    \node at (1,-1) {$\mapsto$};

    \node (a') at (3,0) {$a$};
    \node (a'') at (2,-0.4) {$b$};
    \node (b') at (4,-1) {$a$};
    \node (b'') at (3,-1.4) {$b$};
    \node (c') at (4,-2) {$c$};
    \draw[->] (a') -- (a'');
    \draw[->] (a') -- (b');
    \draw[->] (b') -- (b'');
    \draw[->] (a'') -- (b'');
    \draw[->] (b') -- (c');
    \draw[->] (b'') -- (c');

    \node at (5,-1) {$\mapsto$};

    \node (a') at (7,0) {$a$};
    \node (a'') at (6,-0.4) {$b$};
    \node (b') at (8,-1) {$a$};
    \node (b'') at (7,-1.4) {$b$};
    \node (bb) at (6,-1.4) {$b$};
    \node (c') at (8,-2) {$c$};
    \node (c'') at (7,-2) {$c$};
    \node (cc) at (6,-2) {$c$};
    \draw (a') -- (a'');
    \draw (a') -- (b');
    \draw (b') -- (b'');
    \draw (a'') -- (bb);
    \draw (bb) -- (cc);
    \draw (b') -- (c');
    \draw (b'') -- (c'');
  \end{tikzpicture}
\end{center}
(Here we have also merged the two $c$ leaves from $\qzero\lrangle{0}$ and $\qone\lrangle{0}$.)

In this factorisation, the first tree-to-DAG map is definable as an MSO
transduction. The second operation, that recovers an output tree from a DAG, is
called \emph{unfolding}. In general, all functions computed by top-down tree
transducers belong in the function class
\[ \unfold \circ (\text{MSO tree-to-DAG transductions}) \] which is called
\emph{MSO transductions with sharing} (MSOTS) in~\cite{LambdaTransducer} --- we use this
name here --- and \enquote{MSO term graph transductions}
in~\cite{AttributedMSO}.

\subsection{Attributed vs tree-walking transducers vs MSOTS}

Let us discuss characterisations of MSOTS by machines. There is a recent one
based on transducers using almost linear $\lambda$-terms
in~\cite{KanazawaMSO,LambdaTransducer}. A more classical one relies on
\emph{attributed tree transducers} (ATTs) --- which come from Knuth's attribute
grammars~\cite{Knuth68} --- and on \emph{MSO relabelings}, i.e.\ MSO tree
transductions that keep the structure of the tree intact, and merely change the
label of each node.

\begin{theorem}[{Bloem \& Engelfriet~\cite[Theorem~17]{AttributedMSO}}]
  $\text{ATT} \circ \text{MSO relabeling} = \text{MSOTS}$.
\end{theorem}
(When we write this kind of equality, both sides of $=$ are understood to refer to function classes, and thus the equality denotes an equivalence in expressive power.)

As explained in~\cite[Section~3.2]{EngelfrietPebbleMacro}, the ATT and the
deterministic \emph{tree-walking transducer} (TWT) -- which appears under a
different name in~\cite{EngelfrietPebbleMacro}, cf.\ \Cref{rem:pebbles} -- are
\enquote{essentially notational variations of the same formalism}
(quoting~\cite[Section~8.7]{courcellebook}). Tree-walking transducers are like
top-down tree transducers, except that transitions do not always go
\enquote{downwards} in the input tree: $q\lrangle{\text{some node}\ v\ \text{in
    the input}}$ may be rewritten into an expression that contains
$q'\lrangle{\text{the parent node of}\ v}$. Arguably, TWTs are the natural
generalisation to trees of \emph{two-way transducers} on strings.

The correspondence between the ATT and the TWT relates the former's
\enquote{attributes} to the latter's \enquote{states}, in an analogous fashion
to the bottom-up registers / top-down states correspondence of
\Cref{sec:top-bot}. There are subtleties related to \enquote{noncircularity}
which make the ATT slightly weaker, but they become irrelevant in presence of
preprocessing by MSO relabeling.
\begin{remark}\label{rem:pebbles}
  In~\cite{EngelfrietPebbleMacro}, TWTs are called \enquote{0-pebble tree
    transducers}, as the case $k=0$ of the $k$-pebble tree transducers
  introduced in~\cite{Pebble} (the latter have seen a resurgence in recent years
  in the study of string-to-string polyregular functions~\cite{PolyregSurvey,Kiefer24}). The
  terminology \enquote{tree-walking} was only later used to speak of 0-pebble
  transducers, starting with~\cite{Engelfriet09}, but it was already in use for
  the tree-to-string version of these devices since the late 1960s~\cite{AhoU71}
  --- there, the output strings are produced as if they were unary trees. Let us
  also mention the \enquote{RT(Tree-walk) transducers}
  of~\cite{engelfriet2014contextfree} from the 1980s, discussed
  in~\cite[Section~3.3]{EngelfrietPebbleMacro}.
\end{remark}

When a TWT takes as input a tree that has been preprocessed by an MSO
relabeling, it has access to some extra information on each node; this finite
information recorded by the relabeling is necessarily a regular property of the
input tree with this node distinguished. Equivalently, one could replace the
preprocessing by adding to the TWT the ability to query \enquote{on the fly} the
relevant regular properties of the current node. This feature is usually called
\emph{regular lookaround} or \emph{MSO tests}, see e.g.\ the introduction
to~\cite{EngelfrietIM21}. Hence:
\[ \text{MSOTS} = \text{TWT} \circ \text{MSO relabeling} = \text{TWT with regular lookaround a.k.a.\ MSO tests} \]
However, in the absence of regular lookaround, TWTs are much weaker, since:
\begin{theorem}[{Bojańczyk \& Colcombet}]
  Nondeterministic tree-walking automata do not recognise all regular tree
  languages~\cite{bojanczyk2008tree}, and deterministic TWA are even
  weaker~\cite{BojanczykC06}.
\end{theorem}

Finally, one can impose a \enquote{single use restriction} on both attributed
and (deterministic) tree-walking transducers to avoid sharing phenomena, leading
to characterisations of tree-to-tree MSO transductions (MSOT):
\begin{align*}
  \text{MSOT} &= \text{single-use ATT} \circ \text{MSO relabeling~\cite[Theorem~17]{AttributedMSO}}\\
  &= \text{single-use TWT with MSO tests (regular lookaround)~\cite[Theorem~8.6]{courcellebook}}
\end{align*}

\subsection{Composition hierarchies: MSOTS (\& invisible pebbles / Caucal hierarchy) vs MTT (\& iterated pushdowns / safe higher-order)}

Engelfriet, Hoogeboom and Samwel have characterised the composition of
\emph{two} deterministic tree-walking transducers with MSO tests using a machine
model: the \emph{invisible pebble tree transducer}
(IPTT)~\cite{InvisiblePebbles}, a variant of the $k$-pebble tree transducer of
\Cref{rem:pebbles}. According to the previous subsection, we may equivalently
state their result as:
\begin{theorem}[{rephrasing of~\cite[Theorem~53]{InvisiblePebbles}}]
  $\MSOTS^{2} \mathrel{\overset{\mathrm{def.}}{=}} \MSOTS \circ \MSOTS = \mathrm{IPTT}$.
\end{theorem}

The invisible pebble tree transducer manipulates an unbounded stack of
\enquote{pebbles}, which are pointers to nodes of the input tree. Similarly, the
memory of the \emph{pushdown tree transducer}~\cite{EngelfrietPushdownMacro}
consists mainly of a stack of pointers to input nodes; the difference is that
the reading head of an IPTT moves like that of a TWT, whereas the pointers in a
pushdown tree transducers can only move top-down. This suggests that
$\text{MTT} \subset \text{IPTT}$ which is indeed
true~\cite[Corollary~42]{InvisiblePebbles}.

In fact, we may also look at the composition of $k$ MSOTS beyond the case $k>2$,
as well as the composition of several MTTs. It turns out that the two
hierarchies are interleaved:
\begin{theorem}[{\cite[Corollary~25]{EngelfrietIM21}}]
  $\forall k\geqslant1,\; \MSOTS^{k} \subset \MTT^{k} \subset \MSOTS^{k+1}$.

  In particular, $\displaystyle \bigcup_{k\geqslant1} \MSOTS^{k} = \bigcup_{k\geqslant1} \MTT^{k}$.
\end{theorem}
\begin{proof}[Proof idea]
  This is beause:
  \begin{enumerate}
    \item $\text{MTT} = \text{MSOTS} \circ \text{top-down tree transducer}$~\cite[Lemma~24]{EngelfrietIM21};
    \item $\text{MSOTS} = \text{top-down tree transducer} \circ \text{MSOTS}$~\cite[Theorem~18]{EngelfrietIM21}.
  \end{enumerate}
  (beware: the order of composition is reversed compared to the usual in~\cite{EngelfrietIM21}).
\end{proof}

Let us mention that the main theorem of the above-cited paper~\cite{EngelfrietIM21} is:
\begin{theorem}\label{thm:inaba}
  The functions (cf.~\S\ref{sec:unfolding}) in the $\MSOTS^{k}$
  hierarchy (equivalently, in the $\MTT^{k}$ hierarchy) of \emph{linear
    growth} are exactly the tree-to-tree MSO transductions.
\end{theorem}
Let us also note that:
\[ \text{MSOTS}^{k} = \text{TWT}^{k} \circ \text{MSO relabeling}\]
Since we already know this for $k=1$, the result amounts to eliminating intermediate relabelings in the composition. This can indeed be done, thanks to:
\begin{lemma}
  MSO relabeling $\circ$ MSOTS $=$ MSOTS.
\end{lemma}
\begin{proof}
  According to~\cite[Theorem~18]{EngelfrietIM21}, tree-walking transducers with MSO tests are closed under \emph{post}composition by top-down tree transducers with MSO tests. The latter can compute at least all MSO relabelings, see~\cite[p.~46]{AttributedMSO}.
\end{proof}

Engelfriet and Vogler have shown that the class $\text{MTT}^{k}$ (for some fixed
$k\geqslant1$) corresponds to the expressive power of the \enquote{iterated
  pushdown tree transducer}~\cite{EngelfrietPushdownMacro,EngelfrietHighLevel},
which generalizes the aforementioned pushdown tree transducer by working with a
stack of … of stacks of input pointers, with nesting depth $k$. (See also the
work of Sénizergues on string-to-string
functions~\cite{FerteMarinSenizergues,lvl3}.) Also, it seems natural to believe
that:
\begin{claim}\label{clm:iptt-gen}
  $\text{MSOTS}^{k+1} =$ a variant of iterated pushdown transducers whose
  input pointers can move in a tree-walking rather than top-down fashion, with
  nesting depth $k$, such that the case $k=1$ corresponds to the IPTT.
\end{claim}
\begin{claimproof}[Idea]
  I expect this to be derivable from the
  aforementioned~\cite[Theorem~18]{EngelfrietIM21}, plus the general
  result~\cite[Theorem~5.16]{EngelfrietPushdownMacro} about adding a pushdown
  layer on top of a \enquote{storage type} (see
  also~\cite{engelfriet2014contextfree}), but I'm too lazy to spell out and
  check the details.
\end{claimproof}

The class $\text{MTT}^{k}$ can also be characterised by the \enquote{high level
  tree transducer}~\cite{EngelfrietHighLevel}, which stores some kind of
higher-order functions: the tree contexts in MTTs are first-order i.e.\
tree-to-tree functions, second-order functions map first-order functions to
trees, etc.

\subparagraph{Connections with higher-order recursion schemes.}

Engelfriet and Vogler's high level tree transducer is directly inspired by
Damm's high level grammars~\cite{DammIO}. As explained for instance
in~\cite{BlumOng,ParysHomogeneity}, high level grammars are very close
syntactically --- and equivalently expressive --- to \emph{safe} higher-order
grammars. This notion of safety was introduced by Knapik, Niwiński and
Urzyczyn~\cite{SafeHORS} in the context of \enquote{higher-order recursion
  schemes}; a recursion scheme is a grammar-like specification of a single
infinite tree. Safe recursion schemes of order $k$ describe exactly the same
infinite trees as iterated (a.k.a.\ higher-order) pushdown tree-generating
automata of nesting depth $k$~\cite[Theorems~5.1~and~5.3]{SafeHORS} --- a result
that is strikingly similar to the equivalence between $\text{MTT}^{k}$ and
$k$-iterated pushdown transducers.

Moreover, the same class of infinite trees (for a fixed $k\geqslant0$) can be
equivalently described as those obtained by applying a function in
$(\unfold \circ \text{tree-to-graph MSO transduction})^{k+1}$ to a finite tree
--- this is the \enquote{Caucal hierarchy}~\cite{Caucal02,CarayolW03}. Here the
output of the MSO transductions are rooted directed graphs that are not
necessarily acyclic, and the unfolding operation is generalized to generate
infinite trees when the input has cycles. Since an MSOTS is defined as
$\unfold \circ \text{tree-to-DAG MSOT}$, the Caucal hierarchy is clearly
analogous to the $\text{MSOTS}^{k+1}$ tree transducer hierarchy.

This means that the counterpart to the strict inclusion
$\text{MTT}^{k} \subset \text{MSOTS}^{k+1}$ in the setting of recursion schemes
is an equality. \Cref{clm:iptt-gen} suggests a heuristic explanation: since
iterated pushdown tree-generating automata do not have an input, the difference
between top-down pointers and tree-walking pointers to the input disappears.

\subsection{The MTT composition hierarchy for tree-to-string functions}%
\label{sec:mtt-comp-string}

For a class of tree-to-tree functions $\cC$, let $\SO(\cC)$ be the subclass of functions that output strings encoded as unary trees. Recall from \Cref{sec:string-output} the definition of the yield operation and of the top-down tree-to-string transducer. In that section we explained that
\[ \SO(\MTT^1) = \text{top-down tree-to-string transducer} \qquad\qquad(\text{\Cref{thm:macro-top-down}})\]
For the next levels of the $\text{MTT}^{k}$ hierarchy, we also have:
\begin{proposition}[{rephrasing of~\cite[Theorem~8.7]{EngelfrietHighLevel}}]
  $\forall k \geqslant 1,\; \SO(\MTT^{k+1}) = \yield \circ \MTT^{k}$.
\end{proposition}

The case $k=1$ tells us that $\SO(\MTT^{2}) =$ \enquote{macro tree-to-string transducer} (cf.~\S\ref{sec:macro-string}).
\begin{proof}
  This is established in~\cite{EngelfrietHighLevel} using \enquote{high level
    tree transducers}, but it can also be shown by working directly on
  compositions of MTTs:
  \begin{align*}
    \SO(\MTT^{k+1}) &= \SO(\MTT^1)\circ\MTT^1\circ\MTT^{k-1}\\
    &= \yield \circ \underbrace{(\text{top-down tree transducer}) \circ \MTT}_{=\,\MTT,\ \text{cf.~\cite[Lemma~5]{OutputMacro}}} \circ\; \MTT^{k-1}
  \end{align*}
  (beware: again, in~\cite{OutputMacro}, the notation $\circ$ is flipped compared to ours).
\end{proof}

Thanks to this, for instance, Engelfriet and Maneth's \enquote{bridge
  theorem}~\cite[Theorem~18]{OutputMacro} on output languages, originally stated
using $\yield$, can now be seen as relating $\SO(\MTT^{k})$ and
$\SO(\MTT^{k+1})$ for each $k\geqslant1$. This rephrasing has proved useful in
the recent work~\cite{Sandra}.

\bibliography{bib}

\end{document}